\documentclass[conference]{IEEEtran}
\IEEEoverridecommandlockouts

\usepackage{cite}
\usepackage{amsmath,amssymb,amsfonts}
\usepackage{algorithmic}
\usepackage{graphicx}
\usepackage{textcomp}
\usepackage{xcolor}
\usepackage{amsthm}
\newtheorem{proposition}{Proposition}

\def\BibTeX{{\rm B\kern-.05em{\sc i\kern-.025em b}\kern-.08em
    T\kern-.1667em\lower.7ex\hbox{E}\kern-.125emX}}
\begin{document}

\title{Revisiting the Independence Assumption in LEO Satellite-to-Ground Optical Links: A State-Coupled Joint Fading Model}

\author{
    \IEEEauthorblockN{
        Xinyan Xie\IEEEauthorrefmark{1},
        Xuesong Wang\IEEEauthorrefmark{2},
        Jinghua Zhang\IEEEauthorrefmark{3},
        Fengrui Yang\IEEEauthorrefmark{1},
        Yongheng Wen\IEEEauthorrefmark{1},
        Haoyang He\IEEEauthorrefmark{1},
        Dong Zhao\IEEEauthorrefmark{1}
    }
    \IEEEauthorblockA{
        \IEEEauthorrefmark{1}Fiber Optics Research Center, College of Smart Materials and Future Energy, Fudan University, Shanghai, China
    }
    \IEEEauthorblockA{
        \IEEEauthorrefmark{2}School of Science and Engineering, The Chinese University of Hong Kong, Shenzhen, Guangdong, China
    }
    \IEEEauthorblockA{
        \IEEEauthorrefmark{3}Technology and Engineering Center for Space Utilization, Chinese Academy of Sciences, Beijing, China
    }    
    \IEEEauthorblockA{
        Email: zhaodong@fudan.edu.cn
    }
}

\maketitle

\begin{abstract}
Performance analysis of low Earth orbit (LEO) satellite-to-ground optical links relies on composite fading models that typically evaluate scintillation and angular loss under the assumption of statistical independence. While ensuring analytical tractability, this assumption decouples fading mechanisms driven by the same atmospheric turbulence and fails to capture the distinct effects of free atmosphere (FA) and boundary layer (BL) perturbations. To model this coupling while preserving tractability, this paper develops a state-coupled joint fading model. In the proposed framework, aperture-averaged scintillation and effective angular loss are jointly characterized by a discrete slow atmospheric state, parameterized by separate FA and BL scaling factors. By replacing unconditional independence with state-conditioned independence, the model enables a closed-form derivation of the outage probability, preserving the computational simplicity of the independent baseline. Numerical results show that the independent baseline can misestimate outage under non-nominal layered turbulence states. This outage prediction bias varies with elevation because the relative roles of scintillation and angular loss change with the link geometry, resulting in different residual angular correction requirements for a given outage target.
\end{abstract}


\section{Introduction}
Low Earth orbit (LEO) satellite constellations are becoming an important component of non-terrestrial networks (NTNs) toward 6G~\cite{Intro_1}. Since LEO satellites often return large volumes of data within short contact windows, LEO-to-ground free-space optical (FSO) links are emerging as feeder links with high capacity because of their large optical bandwidth and narrow optical beams~\cite{Intro_2,Intro_3}. However, unlike inter-satellite optical links, LEO-to-ground paths traverse the atmosphere, where absorption, scattering, and turbulence impair the received optical signal~\cite{Intro_4}. Therefore, accurate channel modeling at the physical layer is essential for the analysis and design of LEO-to-ground optical links.

For LEO-to-ground optical links, atmospheric turbulence is a major source of random downlink impairment. Refractive-index fluctuations perturb the received optical field through intensity fluctuations (i.e., scintillation) and wavefront phase distortions, including wavefront tilt associated with angle-of-arrival (AoA) fluctuations. The millisecond scale evolution of these fluctuations induced by the turbulence supports a block-fading approximation, making outage probability a natural reliability metric~\cite{Outage_based}. Recent analytical studies have moved beyond scintillation-only models to composite formulations that include atmospheric fading and pointing errors~\cite{Intro_4,E2E_JSAC}. In such formulations, received power fluctuations are represented by an intensity term associated with scintillation~\cite{PJ} and an angular-loss term associated with angular misalignment. The atmospheric part of this angular loss is induced by AoA fluctuations, while platform jitter adds an independent mechanical component~\cite{JOCN}. Existing analyses typically model the intensity term and the angular-loss term as statistically independent random factors to preserve tractability~\cite{CL}.

This independence assumption enables tractable outage analysis under slow fading~\cite{2007_JLT}, but removes the common turbulence driver of scintillation and atmospheric AoA fluctuations. Both effects are induced by refractive-index fluctuations along the same propagation path, as supported by simultaneous intensity and AoA measurements that yield consistent estimates of the refractive-index structure parameter~\cite{2023_OL}. This common origin does not make scintillation and AoA statistically equivalent because turbulence in the free atmosphere (FA) and the boundary layer (BL) contributes differently to their statistics. Direct derivation of the joint probability density function (PDF) of the resulting intensity and angular-loss terms from the turbulence profile is not tractable for outage analysis. Therefore, the modeling challenge is to construct a tractable joint model that preserves the common turbulence origin of scintillation and atmospheric AoA while retaining distinct statistical roles for the intensity and angular-loss terms.

To address the modeling challenge, we develop a state-coupled joint fading and outage analysis framework for LEO-to-ground optical links. The framework introduces a common slow atmospheric state over a layered turbulence profile. This state jointly controls the scintillation and atmospheric AoA statistics through separate FA and BL scaling, while mechanical jitter remains an independent component of angular loss. The main contributions are summarized as follows:
\begin{itemize}
\item We establish a joint channel model that decomposes received power into deterministic large-scale gain, aperture-averaged scintillation, and effective angular loss. The coupling enters through a discrete slow atmospheric state that controls the scintillation statistics and the atmospheric AoA variance.

\item We derive an analytical outage characterization. By conditioning on the slow atmospheric state and averaging over its state probabilities, the outage probability is expressed as a one-dimensional integral and further reduced to closed form.

\item We quantify the impact of atmospheric state coupling on outage prediction and residual angular correction design. Numerical results show that an independent baseline can misestimate outage under non-nominal layered turbulence states because the relative roles of scintillation and angular loss vary with elevation.
\end{itemize}

\section{System and Channel Model}
\label{sec:II}

We consider a downlink optical link from a LEO satellite to an optical ground station (OGS). Let $P_t$ be the average transmitted optical power and $P_r(t)$ be the instantaneous received optical power. The link geometry is parameterized by the elevation angle $\epsilon(t)\in[\epsilon_{\min},\pi/2]$, where $\epsilon_{\min}$ is the minimum operational elevation angle. As illustrated in Fig.~\ref{fig:1}, the considered LEO-to-OGS downlink channel is decomposed into a large-scale path gain and two random fading components, i.e., aperture-averaged scintillation and effective angular loss induced by AoA fluctuations and mechanical jitter. The received optical power can be expressed as
\begin{equation}
P_r(t)=P_t h_c(t)h_a(t)h_p(t),
\label{eq:channel_decomposition}
\end{equation}
where $h_c(t)$ is the large-scale path gain, $h_a(t)$ is the aperture-averaged scintillation factor, and $h_p(t)$ is the effective angular-loss factor.

\begin{figure}[t]
\centering
\includegraphics[width=0.9\linewidth]{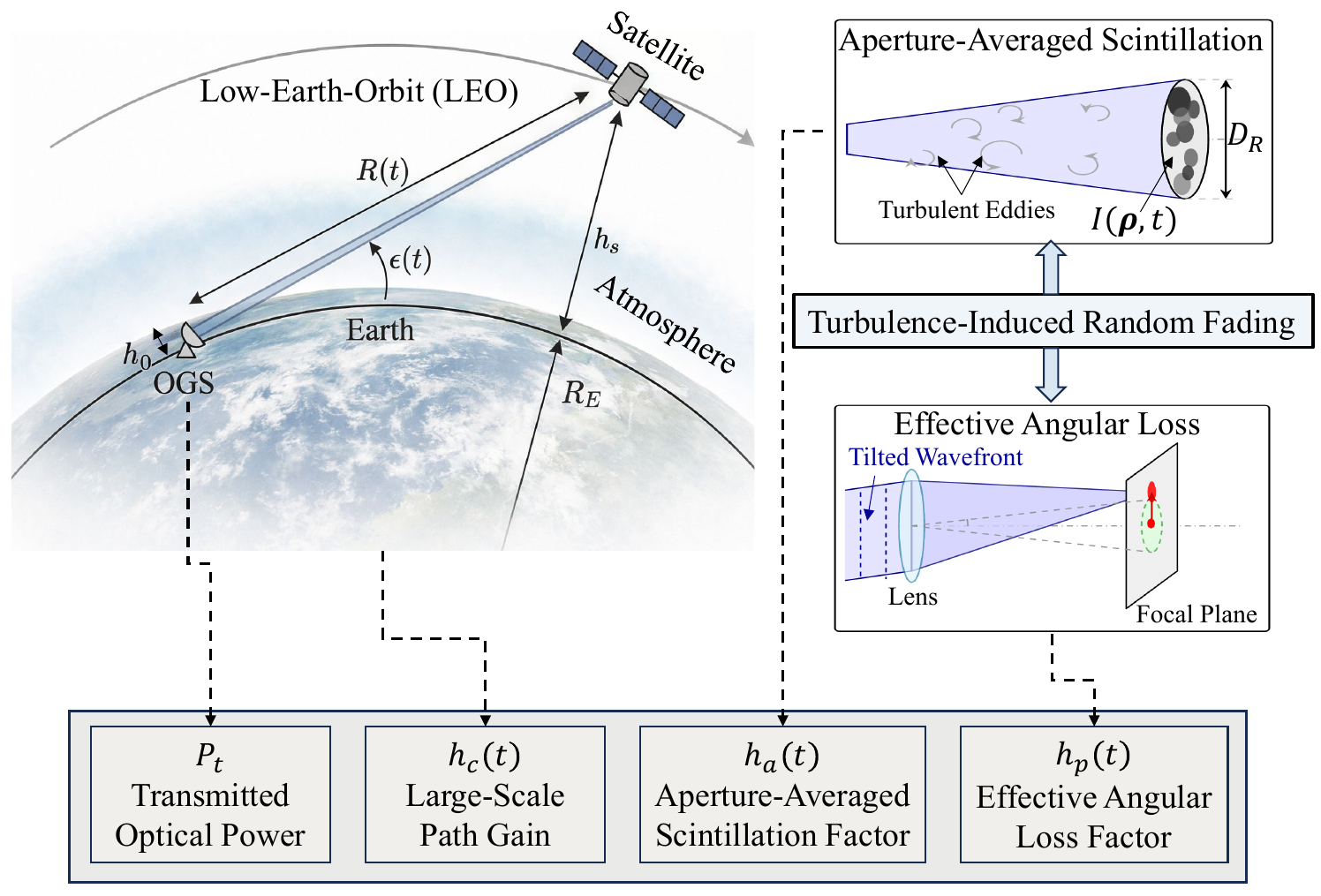}
\caption{LEO-to-OGS downlink system and channel model.}
\label{fig:1}
\end{figure}

\subsection{Large-Scale Path Gain}

The large-scale gain $h_c(t)$ captures the slowly varying deterministic part of the channel, which is governed by the link geometry and average propagation conditions. For a given time instant $t$, it is decomposed into the optical system gain, free-space path loss, and atmospheric extinction, given by
\begin{equation}
h_c(t)
=
\underbrace{\eta_T \eta_R G_T G_R}_{\text{optical system gain}}
\underbrace{\left(\frac{\lambda_0}{4\pi R(t)}\right)^2}_{\text{free-space path loss}}
\underbrace{h_{\mathrm{atm}}(t)}_{\text{atmospheric extinction}},
\label{eq:hc_compact}
\end{equation}
where $G_T \triangleq (\pi D_T/\lambda_0)^2$ and $G_R \triangleq (\pi D_R/\lambda_0)^2$ are the transmitter and receiver telescope gains, respectively. $D_T$ and $D_R$ denote the transmitter and receiver aperture diameters, $\eta_T$ and $\eta_R$ denote the corresponding optical efficiencies, and $\lambda_0$ is the operating wavelength. $R(t)$ is the slant range between the satellite and the OGS. Based on spherical Earth geometry, the slant range at elevation angle $\epsilon(t)$ is given by~\cite{E2E_JSAC}
\begin{equation}
\begin{split}
R(t)
&=
\sqrt{(R_E+h_s)^2-(R_E+h_0)^2\cos^2\epsilon(t)}
\\
&\quad
-(R_E+h_0)\sin\epsilon(t),
\end{split}
\label{eq:slant_range}
\end{equation}
where $R_E$, $h_s$, and $h_0$ denote the Earth radius, the satellite altitude, and the OGS altitude, respectively.

The atmospheric transmittance $h_{\mathrm{atm}}(t)$ represents deterministic extinction along the slant path. Under the standard plane parallel approximation for a thin atmosphere, the factor along the slant path is approximated by $\csc\epsilon(t)$~\cite{Intro_4}, so that
\begin{equation}
h_{\mathrm{atm}}(t)=\exp\!\bigl[-\tau_0(V,\lambda_0)\csc\epsilon(t)\bigr],
\label{eq:hatm_decoupled}
\end{equation}
where $\tau_0(V,\lambda_0)$ is the total zenith optical depth at the operating wavelength $\lambda_0$. For a compact deterministic model, we decompose $\tau_0(V,\lambda_0)$ as $\tau_0(V,\lambda_0)=\tau_{\mathrm{abs}}(\lambda_0)+\tau_{\mathrm{scat}}(V,\lambda_0)$, where $\tau_{\rm abs}(\lambda_0)=-\ln T_{\mathrm{abs}}(\lambda_0)$, and $T_{\mathrm{abs}}(\lambda_0)$ denotes the absorption-only zenith transmittance~\cite{1990_T}. $\tau_{\mathrm{scat}}(V,\lambda_0)$ is given by~\cite{2001_kruse}
\begin{equation}
\tau_{\mathrm{scat}}(V,\lambda_0)
=
\frac{3.912}{V}
\left(\frac{\lambda_0}{550\,\mathrm{nm}}\right)^{-q(V)}
H_{\mathrm{BL}},
\label{eq:tau_scat}
\end{equation}
where $V$ is the meteorological visibility, $q(V)$ is the Kruse exponent, and $H_{\mathrm{BL}}$ is the equivalent boundary-layer scale height. Therefore, the optical system gain is constant, while the free-space path loss and $h_{\mathrm{atm}}(t)$ vary deterministically with the pass geometry through $R(t)$ and $\epsilon(t)$. Hence, $h_c(t)$ is a deterministic large-scale term and will be absorbed into the normalized threshold in Section~\ref{sec:III}.

\subsection{Turbulence-Induced Random Fading}

With the deterministic large-scale gain $h_c(t)$ specified, the residual fluctuation is induced by atmospheric turbulence. For the statistical closure of the two random components in \eqref{eq:channel_decomposition}, we adopt the aperture-averaged scintillation index $\sigma_a^2(t)$ and the mean-square AoA fluctuation $\beta_{\mathrm{rms}}^2(t)$ as the canonical second-order descriptors of the amplitude term $h_a(t)$ and the angular term $h_p(t)$. Both are generated from the refractive-index structure parameter profile $C_n^2(h)$.

The standard Hufnagel-Valley (HV) profile is given by~\cite{Cn}
\begin{equation}
\begin{aligned}
C_n^2(h)
&=
0.00594\!\left(\frac{v_{\mathrm{rms}}}{27}\right)^2
(10^{-5}h)^{10}e^{-\frac{h}{1000}}
\\
&\quad+
2.7\times10^{-16}e^{-\frac{h}{1500}}
+
A_0e^{-\frac{h-h_0}{100}},
\end{aligned}
\label{eq:Cn_HV}
\end{equation}
where $h$ is the altitude, $v_{\mathrm{rms}}$ is the high-altitude root
mean square wind speed, and $A_0$ is the turbulence strength at the OGS altitude $h_0$. For the proposed state-coupled model, we decompose \eqref{eq:Cn_HV} into a FA component $C_{n,\mathrm{FA}}^2(h)$, formed by the first two terms, and a BL component $C_{n,\mathrm{BL}}^2(h)$, formed by the last term. The rationale is that the first two terms are independent of $A_0$ and decay over kilometer scales, whereas the last is governed by $A_0$ and decays over the 100 m scale. Accordingly, \eqref{eq:Cn_HV} is rewritten as
\begin{equation}
C_n^2(h)=C_{n,\mathrm{FA}}^2(h)+C_{n,\mathrm{BL}}^2(h),
\label{eq:Cn_split}
\end{equation}
and the corresponding second-order statistics are given by
\begin{equation}
\mathbf z(t)\triangleq
\begin{bmatrix}
\sigma_a^2(t)\\
\beta_{\mathrm{rms}}^2(t)
\end{bmatrix}
=
\int_{h_0}^{h_s}\mathbf K(h,t)\,C_n^2(h)\,dh,
\label{eq:z_operator}
\end{equation}
where $\mathbf K(h,t)\triangleq[\mathcal K_a(h,t),\,\mathcal K_\beta(t)]^\top$. The two kernels are given by~\cite{PJ,Cn}
\begin{subequations}\label{eq:kernels}
\begin{align}
\mathcal K_a(h,t)
&=
8.70\,k_0^{\frac{7}{6}}\Delta h^{\frac{5}{6}}\csc^{\frac{11}{6}}\!\epsilon(t)
\notag\\
&\quad\times
\Re\!\left\{[\alpha(t)+i\xi(h)]^{\frac{5}{6}}-\alpha^{\frac{5}{6}}(t)\right\},
\label{eq:Ka}
\\
\mathcal K_\beta(t)
&=
2.91\,D_R^{-\frac{1}{3}}\csc\epsilon(t),
\label{eq:Kb}
\end{align}
\end{subequations}
where $k_0=2\pi/\lambda_0$, $\Delta h=h_s-h_0$, $\xi(h)=(h-h_0)/\Delta h$, and $\alpha(t)\triangleq k_0D_R^2/(16\Delta h\csc\epsilon(t))$. Since \eqref{eq:z_operator} is linear in $C_n^2(h)$, the decomposition in \eqref{eq:Cn_split} induces
\begin{equation}
\mathbf z(t)=\mathbf z_{\mathrm{FA}}(t)+\mathbf z_{\mathrm{BL}}(t),
\label{eq:z_split}
\end{equation}
where $\mathbf z_j(t)\triangleq[\sigma_{a,j}^2(t),\,\beta_{\mathrm{rms},j}^2(t)]^\top$ for $j\in\{\mathrm{FA},\mathrm{BL}\}$, with
\begin{subequations}\label{eq:component_stats}
\begin{align}
\sigma_{a,j}^2(t)
&\triangleq
\int_{h_0}^{h_s}\mathcal K_a(h,t)\,C_{n,j}^2(h)\,dh,
\label{eq:sigma_component}
\\
\beta_{\mathrm{rms},j}^2(t)
&\triangleq
\int_{h_0}^{h_s}\mathcal K_\beta(t)\,C_{n,j}^2(h)\,dh.
\label{eq:beta_component}
\end{align}
\end{subequations}
The decomposition in \eqref{eq:Cn_split} is applied to the turbulence profile, while Proposition~\ref{prop:kernel_asymptotics} describes how the FA and BL components are weighted by the scintillation and AoA kernels.

\begin{proposition}
\label{prop:kernel_asymptotics}
For fixed $t$, $\mathcal K_\beta(t)$ is independent of $h$, whereas
\begin{equation}
\mathcal K_a(h,t)=C_a(t)(h-h_0)^2+\mathcal O\!\bigl((h-h_0)^4\bigr),
\qquad h\to h_0,
\label{eq:Ka_asym}
\end{equation}
with
\begin{equation}
C_a(t)=
\frac{29}{48}\,
k_0^{\frac{7}{6}}\csc^{\frac{11}{6}}\!\epsilon(t)\,
\alpha^{-\frac{7}{6}}(t)\,\Delta h^{-\frac{7}{6}}.
\label{eq:Ca_coeff}
\end{equation}
Hence $\mathcal K_a(h,t)$ and $\mathcal K_\beta(t)$ cannot be represented by a common scalar multiple.
\end{proposition}

\begin{IEEEproof}
The $h$-independence of $\mathcal K_\beta(t)$ follows directly from \eqref{eq:Kb}. For $\mathcal K_a(h,t)$, let $\xi=\xi(h)$. As $h\to h_0$, one has $\xi\to 0$, and
\begin{equation}
\left(1+i\frac{\xi}{\alpha}\right)^{\frac{5}{6}}
=
1+\frac{5i}{6}\frac{\xi}{\alpha}
+\frac{5}{72}\frac{\xi^2}{\alpha^2}
+\mathcal O(\xi^3).
\label{eq:binom_expand}
\end{equation}
Since the linear term is purely imaginary,
\begin{equation}
\Re\!\left\{(\alpha+i\xi)^{\frac{5}{6}}-\alpha^{\frac{5}{6}}\right\}
=
\frac{5}{72}\alpha^{-\frac{7}{6}}\xi^2+\mathcal O(\xi^4).
\label{eq:real_expand}
\end{equation}
Substituting \eqref{eq:real_expand} into \eqref{eq:Ka} and using $\xi(h)=(h-h_0)/\Delta h$ yields \eqref{eq:Ka_asym} and \eqref{eq:Ca_coeff}. The last claim follows because $\mathcal K_\beta(t)$ is constant in $h$, whereas $\mathcal K_a(h,t)$ vanishes quadratically as $h\to h_0$.
\end{IEEEproof}

Proposition~\ref{prop:kernel_asymptotics} shows that near-ground turbulence is suppressed in $\mathcal K_a(h,t)$ but not in $\mathcal K_\beta(t)$. A single scalar perturbation of $C_n^2(h)$ therefore cannot induce proportional variations in $\sigma_a^2(t)$ and $\beta_{\mathrm{rms}}^2(t)$, which motivates the FA/BL decomposition and the state parameterization in Section~\ref{sec:II-C}.

\subsubsection{Aperture-Averaged Scintillation}

Using the aperture-averaged scintillation index $\sigma_a^2(t)$ obtained in \eqref{eq:sigma_component}, we model the scintillation by a unit-mean Gamma distribution, which is a suitable approximation for downlinks with aperture averaging and a small scintillation index~\cite{E2E_JSAC}. Accordingly, we define
\begin{equation}
h_a(t)\triangleq \frac{P_a(t)}{\mathbb E[P_a(t)]},
\label{eq:ha_def}
\end{equation}
where $P_a(t)$ denotes the aperture-collected optical power affected by scintillation. Since $\sigma_a^2(t)$ is the normalized variance of $P_a(t)$, the normalized scintillation factor satisfies $\mathbb E[h_a(t)]=1$ and $\mathrm{Var}[h_a(t)]=\sigma_a^2(t)$. With $m_a(t)\triangleq 1/\sigma_a^2(t)$, the corresponding PDF is
\begin{equation}
f_{h_a}(x|t)=
\frac{[m_a(t)]^{m_a(t)}}{\Gamma(m_a(t))}
x^{m_a(t)-1}e^{-m_a(t)x},
\qquad x>0,
\label{eq:ha_pdf}
\end{equation}
where $\Gamma(\cdot)$ is the Gamma function.

\subsubsection{Effective Angular Loss}

Using the AoA statistic $\beta_{\mathrm{rms}}^2(t)$ in \eqref{eq:beta_component}, we next model the effective angular loss. In the downlink, tip-tilt correction reduces but does not eliminate the AoA fluctuation, so we introduce a residual angular correction factor $\eta_{\mathrm{tt}}\in(0,1]$, where $\eta_{\mathrm{tt}}=1$ corresponds to no correction. Let $\boldsymbol{\theta}_a(t)$ denote the residual atmospheric angular error. Under the zero-mean 2-D isotropic Gaussian assumption, its mean-square magnitude is $\eta_{\mathrm{tt}}\beta_{\mathrm{rms}}^2(t)$, so the corresponding per-axis variance is
\begin{equation}
\sigma_{\theta_a}^2(t)=\frac{1}{2}\eta_{\mathrm{tt}}\beta_{\mathrm{rms}}^2(t).
\label{eq:sigma_theta_a}
\end{equation}
We then define the total angular misalignment as
\begin{equation}
\boldsymbol{\theta}(t)=\boldsymbol{\theta}_a(t)+\boldsymbol{\theta}_m(t),
\label{eq:theta_sum}
\end{equation}
where $\boldsymbol{\theta}_a(t)\sim\mathcal N(\mathbf 0,\sigma_{\theta_a}^2(t)\mathbf I_2)$ and $\boldsymbol{\theta}_m(t)\sim\mathcal N(\mathbf 0,\sigma_m^2\mathbf I_2)$ denotes an independent mechanical-jitter component. It follows that $\boldsymbol{\theta}(t)\sim\mathcal N(\mathbf 0,\sigma_\theta^2(t)\mathbf I_2)$ with $\sigma_\theta^2(t)=\sigma_m^2+\sigma_{\theta_a}^2(t)$. Since the norm of a zero-mean 2-D isotropic Gaussian vector is Rayleigh distributed, the effective angular-loss factor is modeled as
\begin{equation}
h_p(t)=\exp\!\left[-\frac{2\|\boldsymbol{\theta}(t)\|^2}{\theta_{\mathrm{eq}}^2}\right],
\label{eq:hp_def}
\end{equation}
where $\theta_{\mathrm{eq}}$ is an effective angular acceptance parameter, defined as $\theta_{\mathrm{eq}}^2\triangleq \theta_{\mathrm{FOV}}^2+(\lambda_0/D_R)^2$. This parameter approximates the receiver acceptance angle by combining the FOV contribution and the angular spread set by diffraction under a Gaussian acceptance model. The resulting density is
\begin{equation}
f_{h_p}(u|t)=q_p(t)u^{q_p(t)-1},
\qquad 0<u\le 1,
\label{eq:hp_pdf}
\end{equation}
with $q_p(t)\triangleq \theta_{\mathrm{eq}}^2/(4\sigma_\theta^2(t))$.

\subsection{State-Coupled Joint Fading Model}
\label{sec:II-C}
Conventional FSO analyses typically impose unconditional independence, i.e., $f_{h_a,h_p}(x,u|t)=f_{h_a}(x|t)f_{h_p}(u|t)$. Since deriving the joint PDF directly from the continuous turbulence profile is not tractable for outage analysis, we represent the slow atmospheric variability by a finite mixture over discrete states. Let $S(t)\in\mathcal{W}\triangleq\{1,2,\ldots,W\}$ denote the discrete atmospheric state, with $\pi_s(t)\triangleq\Pr[S(t)=s]$, $\pi_s(t)\ge0$, and $\sum_{s\in\mathcal{W}}\pi_s(t)=1$. For each state $s$, define a scaling vector for the FA and BL components as $\boldsymbol{\chi}_s\triangleq[\chi_{s,\mathrm{FA}},\,\chi_{s,\mathrm{BL}}]^\top$. The corresponding turbulence profile conditioned on state $s$ is
\begin{equation}
C_{n,s}^2(h)=\chi_{s,\mathrm{FA}}C_{n,\mathrm{FA}}^2(h)+\chi_{s,\mathrm{BL}}C_{n,\mathrm{BL}}^2(h).
\end{equation}

Because \eqref{eq:component_stats} is linear in $C_n^2(h)$, this profile determines the second-order statistics under state $s$ as
\begin{equation}
\mathbf z_s(t)\triangleq
\begin{bmatrix}
\sigma_{a,s}^2(t)\\
\beta_{\mathrm{rms},s}^2(t)
\end{bmatrix}
=
\underbrace{
\begin{bmatrix}
\sigma_{a,\mathrm{FA}}^2(t) & \sigma_{a,\mathrm{BL}}^2(t)\\
\beta_{\mathrm{rms},\mathrm{FA}}^2(t) & \beta_{\mathrm{rms},\mathrm{BL}}^2(t)
\end{bmatrix}}_{\mathbf M(t)}
\boldsymbol{\chi}_s.
\label{eq:z_state}
\end{equation}
Given the slow atmospheric state $S(t)=s$, the remaining fast fluctuations on the millisecond scale are assumed to be conditionally independent:
\begin{equation}
f_{h_a,h_p|S}(x,u|s,t)=f_{h_a|S}(x|s,t)f_{h_p|S}(u|s,t).
\label{eq:conditional_ind}
\end{equation}
The conditional marginals retain the parametric forms of \eqref{eq:ha_pdf} and \eqref{eq:hp_pdf}, with parameters $m_{a,s}(t)=1/\sigma_{a,s}^2(t)$ and $q_{p,s}(t)=\theta_{\mathrm{eq}}^2/\!\bigl(4\sigma_m^2+2\eta_{\mathrm{tt}}\beta_{\mathrm{rms},s}^2(t)\bigr)$. Averaging over the state space gives the unconditional joint PDF as
\begin{equation}
f_{h_a,h_p}(x,u|t)=
\sum_{s\in\mathcal W}\pi_s(t)\,
f_{h_a|S}(x|s,t)\,
f_{h_p|S}(u|s,t).
\label{eq:joint_uncond}
\end{equation}

Unlike the independent baseline, this mixture distribution is generally not factorizable into the product of two marginal PDFs. Thus, $h_a(t)$ and $h_p(t)$ are coupled through the shared slow state $S(t)$, while remaining conditionally independent within each state.

\section{Performance Analysis}
\label{sec:III}

We characterize the proposed joint fading model with atmospheric state coupling through the instantaneous outage probability. For a fixed time $t$, an outage occurs when the received optical power is below the threshold $P_{\mathrm{th}}$, i.e.,
$P_t h_c(t)h_a(t)h_p(t)<P_{\mathrm{th}}$.
Because $h_c(t)$ is deterministic at a given $t$, we define the normalized threshold as
\begin{equation}
\nu(t)\triangleq \frac{P_{\mathrm{th}}}{P_t h_c(t)},
\label{eq:nu_def}
\end{equation}
and the outage event becomes $h_a(t)h_p(t)<\nu(t)$.

From \eqref{eq:joint_uncond}, the unconditional outage probability is
\begin{equation}
P_{\mathrm{out}}^{\mathrm{cpl}}(t)
=
\sum_{s\in\mathcal W}\pi_s(t)\,P_{\mathrm{out},s}(t),
\label{eq:pout_mix}
\end{equation}
where
$P_{\mathrm{out},s}(t)\triangleq
\Pr[h_a(t)h_p(t)<\nu(t)\mid S(t)=s]$.
Using the conditional independence in \eqref{eq:conditional_ind}, the outage probability conditioned on state $s$ is
\begin{equation}
P_{\mathrm{out},s}(t)=
\int_0^1
F_{h_a|S}\!\left(\frac{\nu(t)}{u}\Bigm|s,t\right)
f_{h_p|S}(u|s,t)\,du.
\label{eq:pout_cond_int}
\end{equation}

Let $m_s\triangleq m_{a,s}(t)$ and $q_s\triangleq q_{p,s}(t)$.
Since $h_a(t)\mid S(t)=s$ follows the Gamma distribution in \eqref{eq:ha_pdf}, its conditional cumulative distribution function (CDF) is
\begin{equation}
F_{h_a|S}\!\left(\frac{\nu(t)}{u}\Bigm|s,t\right)
=
\frac{\gamma\!\left(m_s,\frac{m_s\nu(t)}{u}\right)}{\Gamma(m_s)},
\label{eq:cond_cdf}
\end{equation}
where $\gamma(\cdot,\cdot)$ denotes the lower incomplete Gamma function. Substituting \eqref{eq:cond_cdf} and \eqref{eq:hp_pdf} into \eqref{eq:pout_cond_int} gives
\begin{equation}
P_{\mathrm{out},s}(t)
=
\frac{q_s}{\Gamma(m_s)}
\int_0^1
\gamma\!\left(m_s,\frac{m_s\nu(t)}{u}\right)
u^{q_s-1}\,du.
\label{eq:pout_1d}
\end{equation}

\begin{proposition}
For $m_s>0$ and $q_s>0$, the conditional outage probability in \eqref{eq:pout_1d} is
\begin{equation}
\begin{aligned}
P_{\mathrm{out},s}(t)
&=
\frac{1}{\Gamma(m_s)}
\Bigl[
\gamma\!\left(m_s,m_s\nu(t)\right) \\
&\qquad+
\left(m_s\nu(t)\right)^{q_s}
\Gamma\!\left(m_s-q_s,m_s\nu(t)\right)
\Bigr],
\end{aligned}
\label{eq:pout_cond_closed}
\end{equation}
where $\Gamma(\cdot,\cdot)$ denotes the upper incomplete Gamma function. For $m_s\le q_s$, $\Gamma(m_s-q_s,m_s\nu(t))$ is evaluated by its standard extension, which is finite for $m_s\nu(t)>0$.
\end{proposition}

\begin{IEEEproof}
The proof is given in Appendix.
\end{IEEEproof}

\section{Numerical Results}
In this section, numerical results are provided to assess the proposed state-coupled model against the conventional independent baseline. Unless otherwise stated, the system parameters are listed in Table~\ref{tab:system_parameters}. For comparison normalization, the outage threshold $P_{\mathrm{th}}$ is set by a reference operating point such that the independent baseline satisfies $P_{\mathrm{out}}=10^{-2}$ at $\epsilon_{\mathrm{ref}}=25^\circ$. The same threshold is then used in all subsequent results.

\subsection{Second-Order Statistics Analysis}

As shown in \eqref{eq:z_state}, the atmospheric state determines the two second-order descriptors through the FA and BL scaling factors. We consider three atmospheric cases, $c\in\{\mathrm{nom},\mathrm{BL},\mathrm{FA}\}$, corresponding to the nominal, BL-dominant, and FA-dominant regimes. For state $s$ in case $c$, let $\mathbf z_{s,c}(t)$ denote the corresponding vector of second-order statistics. The state-averaged statistics are defined as
\begin{equation}
\bar{\mathbf z}_c(t)\triangleq
[\bar{\sigma}_{a,c}^2(t),\bar{\beta}_{\mathrm{rms},c}^2(t)]^\top
=
\sum_{s\in\mathcal W}\pi_s\,\mathbf z_{s,c}(t).
\label{eq:case_avg_stats}
\end{equation}
The FA and BL scaling factors below are normalized scenario multipliers applied to the baseline HV profile in \eqref{eq:Cn_HV}, rather than universal atmospheric constants. We use three states, $\mathcal W=\{1,2,3\}$, with fixed probabilities $\{\pi_s\}_{s=1}^3=[0.4,0.4,0.2]$ for low, moderate, and strong slow atmospheric states. The nominal case uses $\{\chi_{s,\mathrm{FA}}\}_{s=1}^3=[0.7,1.0,1.3]$ and $\{\chi_{s,\mathrm{BL}}\}_{s=1}^3=[0.7,1.0,1.5]$. The BL-dominant case uses $\{\chi_{s,\mathrm{FA}}\}_{s=1}^3=[1.0,1.4,1.8]$ and $\{\chi_{s,\mathrm{BL}}\}_{s=1}^3=[1.0,2.0,3.5]$, whereas the FA-dominant case uses $\{\chi_{s,\mathrm{FA}}\}_{s=1}^3=[1.0,1.5,2.5]$ and $\{\chi_{s,\mathrm{BL}}\}_{s=1}^3=[1.0,1.2,1.5]$. With measured or reanalysis turbulence profiles, the state probabilities and scaling factors can be calibrated for a specific site and observation period.

As shown in Fig.~\ref{fig:2}, both $\bar{\sigma}_{a,c}^2$ and $\bar{\beta}_{\mathrm{rms},c}^2$ decrease with elevation because a higher elevation shortens the atmospheric propagation path. However, the two statistics exhibit different case orderings. Over the considered elevation range, the FA-dominant case gives the largest $\bar{\sigma}_{a,c}^2$, whereas the BL-dominant case gives the largest $\bar{\beta}_{\mathrm{rms},c}^2$. This difference follows from the kernels in \eqref{eq:kernels} and the FA/BL scaling in \eqref{eq:z_state}, which weight the layered turbulence perturbations differently for scintillation and AoA. For example, at $\epsilon=25^\circ$, the nominal case gives $(\bar{\sigma}_{a,\mathrm{nom}}^2,\bar{\beta}_{\mathrm{rms},\mathrm{nom}}^2)=(0.0494,2.29\times10^{-11})$. The BL-dominant case gives $(0.0672,4.04\times10^{-11})$, while the FA-dominant case gives $(0.0763,2.88\times10^{-11})$. These results indicate that a single scalar perturbation of $C_n^2(h)$ is insufficient to represent layered slow atmospheric variability.

\subsection{Outage Probability Comparison}

\begin{table}[t]
\caption{System parameters}
\label{tab:system_parameters}
\centering
\footnotesize
\setlength{\tabcolsep}{4pt}
\renewcommand{\arraystretch}{1.05}
\begin{tabular}{llll}
\hline
\textbf{Parameter} & \textbf{Value} & \textbf{Parameter} & \textbf{Value} \\
\hline
$R_E$ & $6371~\mathrm{km}$ &
$h_s$ & $550~\mathrm{km}$ \\

$h_0$ & $20~\mathrm{m}$ &
$\lambda_0$ & $1550~\mathrm{nm}$ \\

$P_t$ & $2~\mathrm{W}$ &
$D_T$ & $0.10~\mathrm{m}$ \\

$D_R$ & $0.30~\mathrm{m}$ &
$\eta_T, \eta_R$ & $0.80$ \\

$V$ & $15~\mathrm{km}$ &
$q(V)$ & $1.3$ \\

$H_{\mathrm{BL}}$ & $2~\mathrm{km}$ &
$T_{\mathrm{abs}}(\lambda_0)$ & $e^{-0.05}$ \\

$A_0$ & $1.7\times10^{-14}~\mathrm{m}^{-2/3}$ &
$v_{\mathrm{rms}}$ & $21~\mathrm{m/s}$ \\

$\theta_{\mathrm{FOV}}$ & $20~\mu\mathrm{rad}$ &
$\sigma_m$ & $2~\mu\mathrm{rad}$ \\
\hline
\end{tabular}
\end{table}

\begin{figure}[t]
\centering
\includegraphics[width=0.9\linewidth]{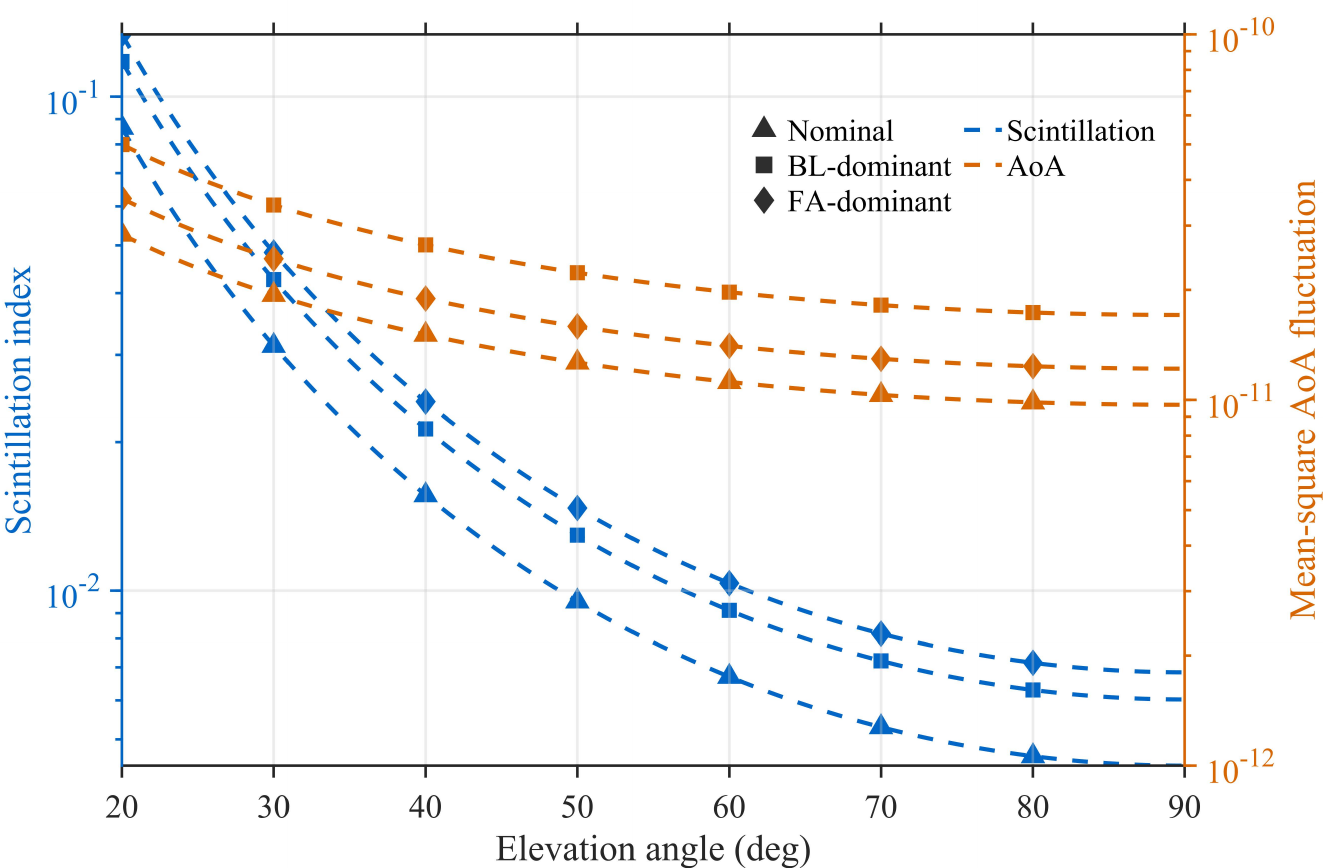}
\caption{State-averaged scintillation index and mean-square AoA fluctuation versus elevation under the nominal, BL-dominant, and FA-dominant atmospheric cases.}
\label{fig:2}
\end{figure}

With $P_{\mathrm{th}}$ fixed by the reference operating point, Fig.~\ref{fig:3} shows how the second-order statistics in Fig.~\ref{fig:2} are mapped into outage behavior. At $\epsilon=25^\circ$, the independent baseline gives the prescribed outage probability $10^{-2}$. The nominal, BL-dominant, and FA-dominant cases give $1.04\times10^{-2}$, $2.85\times10^{-2}$, and $3.30\times10^{-2}$, respectively. This result shows that the independent baseline is nearly aligned with the nominal case at the reference point, but can underestimate outage under the considered non-nominal atmospheric cases.

The ordering of the coupled cases changes with elevation. At $\epsilon=30^\circ$, the FA-dominant case gives the largest outage, $9.22\times 10^{-4}$, compared with $4.83\times 10^{-4}$ for the BL-dominant case and $1.60\times 10^{-5}$ for the independent baseline. At $\epsilon=40^\circ$, the ordering between the two coupled cases reverses, with the BL-dominant case giving $4.16\times 10^{-8}$ and the FA-dominant case giving $1.54\times 10^{-8}$. This reversal suggests that scintillation has a stronger influence at lower elevations, whereas angular loss becomes relatively more influential as the elevation increases. The discrepancy between the proposed model and the independent baseline is therefore not a fixed outage offset, but is governed by the elevation-dependent balance between scintillation and angular loss. The inset compares the analytical curves with Monte Carlo estimates in the resolvable outage range using $10^6$ realizations. Within the plotted low-elevation range where the empirical outage is above $10^{-4}$, the maximum relative error is below 4.97\%.

\subsection{Residual Angular Correction Requirement}

\begin{figure}[t]
\centering
\includegraphics[width=0.9\linewidth]{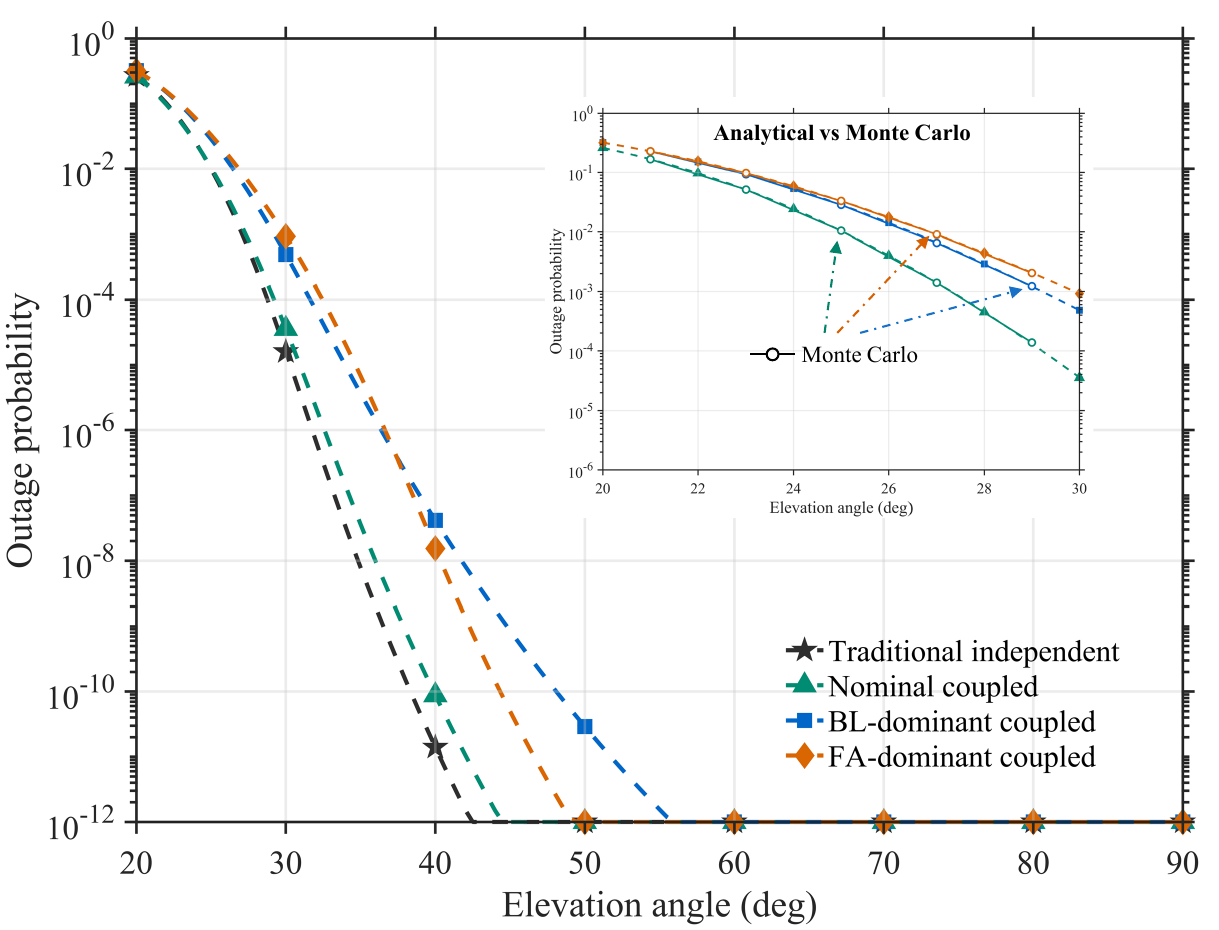}
\caption{Outage probability versus elevation for the independent baseline and the proposed state-coupled model. The inset compares analytical and Monte Carlo results in the resolvable outage range.}
\label{fig:3}
\end{figure}

\begin{figure}[t]
\centering
\includegraphics[width=0.9\linewidth]{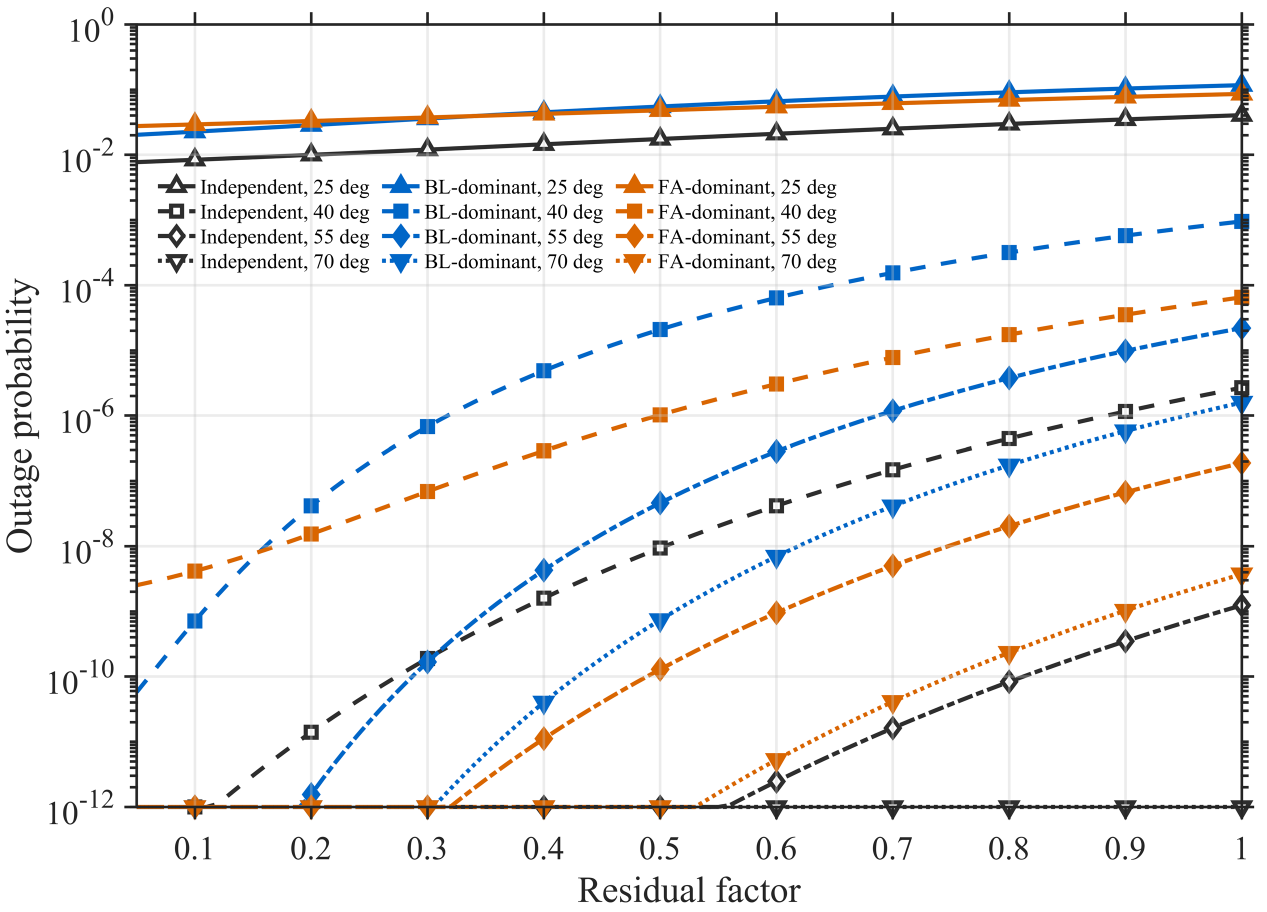}
\caption{Outage probability versus residual angular correction factor for the independent baseline and the proposed state-coupled model under BL-dominant and FA-dominant atmospheric cases.}
\label{fig:4}
\end{figure}

Fig.~\ref{fig:4} compares the outage probability versus the residual angular correction factor $\eta_{\mathrm{tt}}$ for the BL-dominant and FA-dominant cases at $\epsilon=25^\circ$, $40^\circ$, $55^\circ$, and $70^\circ$. A smaller $\eta_{\mathrm{tt}}$ represents stronger suppression of residual AoA fluctuations. According to \eqref{eq:sigma_theta_a} and \eqref{eq:hp_pdf}, $\eta_{\mathrm{tt}}$ affects outage through the residual AoA variance and the resulting angular-loss distribution. The separation between the BL-dominant and FA-dominant curves follows the ordering of $\bar{\beta}_{\mathrm{rms},c}^2$ in Fig.~\ref{fig:2}.

At $\epsilon=25^\circ$, the independent baseline reaches $10^{-2}$ at $\eta_{\mathrm{tt}}=0.2$, whereas both coupled cases remain above $10^{-2}$ over $0.05\leq\eta_{\mathrm{tt}}\leq1$. At $\epsilon=40^\circ$, $55^\circ$, and $70^\circ$, the BL-dominant case gives larger outage than the FA-dominant case for the same $\eta_{\mathrm{tt}}$, consistent with $\bar{\beta}_{\mathrm{rms},\mathrm{BL}}^2>\bar{\beta}_{\mathrm{rms},\mathrm{FA}}^2$ in Fig.~\ref{fig:2}. For a fixed outage target, the BL-dominant case therefore requires a smaller residual angular correction factor. The independent baseline removes this atmospheric dependence and can bias the residual angular correction requirement in link design.

\section{Conclusion}
This paper revisited outage analysis for LEO-to-ground optical links by considering the common turbulence origin of scintillation after aperture averaging and angular loss induced by AoA fluctuations. Rather than treating the two fading factors as unconditionally independent, the proposed model introduced a discrete slow atmospheric state to parameterize the FA and BL turbulence components. Conditioned on this state, the marginal fading models remain tractable, while averaging over the state space yields a joint distribution that cannot be factorized. The resulting outage probability was expressed as a one-dimensional integral and further reduced to closed form. Numerical results showed that the independent baseline can misestimate outage under non-nominal layered turbulence states because the relative roles of scintillation and angular loss vary with elevation. These variations lead to different residual angular correction requirements for a given outage target. Therefore, outage prediction and angular correction design for LEO optical downlinks should account for atmospheric state coupling rather than rely solely on unconditional independence.

\section*{Acknowledgment}
This work was supported by the Science and Technology Commission of Shanghai Municipality (17DZ2280600).

\appendix
\label{app:prop1}

Let \(\nu\triangleq \nu(t)\). From \eqref{eq:pout_1d}, by setting \(z=m_s\nu/u\), one has \(u=m_s\nu/z\) and \(du=-(m_s\nu)z^{-2}dz\). Thus, we have
\begin{equation}
P_{\mathrm{out},s}(t)
=
\frac{q_s(m_s\nu)^{q_s}}{\Gamma(m_s)}
\int_{m_s\nu}^{\infty}
\gamma(m_s,z)\,z^{-q_s-1}\,dz.
\label{eq:app_subst}
\end{equation}

Next, applying integration by parts to \eqref{eq:app_subst} with
\(U(z)=\gamma(m_s,z)\) and
\(dV=q_s(m_s\nu)^{q_s}z^{-q_s-1}dz/\Gamma(m_s)\), it follows that
\(dU=z^{m_s-1}e^{-z}dz\) and
\(V=-(m_s\nu)^{q_s}z^{-q_s}/\Gamma(m_s)\). ~\eqref{eq:app_subst} can be expanded as
\begin{equation}
\begin{aligned}
P_{\mathrm{out},s}(t)
&=
\left[
-\frac{(m_s\nu)^{q_s}}{\Gamma(m_s)}
\gamma(m_s,z)\,z^{-q_s}
\right]_{m_s\nu}^{\infty} \\
&\quad+
\frac{(m_s\nu)^{q_s}}{\Gamma(m_s)}
\int_{m_s\nu}^{\infty}
z^{m_s-q_s-1}e^{-z}\,dz .
\end{aligned}
\label{eq:app_parts}
\end{equation}

Since \(q_s>0\), the boundary term in \eqref{eq:app_parts} vanishes as \(z\to\infty\). At the lower limit \(z=m_s\nu\), it reduces to \(\gamma(m_s,m_s\nu)/\Gamma(m_s)\). The remaining integral is recognized as the upper incomplete Gamma function, i.e.,
\begin{equation}
\int_{m_s\nu}^{\infty}
z^{m_s-q_s-1}e^{-z}\,dz
=
\Gamma(m_s-q_s,m_s\nu).
\label{eq:app_gamma}
\end{equation}
Substituting these results into \eqref{eq:app_parts} yields
\begin{equation}
P_{\mathrm{out},s}(t)
=
\frac{\gamma(m_s,m_s\nu)}{\Gamma(m_s)}
+
\frac{(m_s\nu)^{q_s}\Gamma(m_s-q_s,m_s\nu)}{\Gamma(m_s)}.
\end{equation}
This proves \eqref{eq:pout_cond_closed}, and the case $m_s=q_s$ follows by continuity.

\bibliographystyle{IEEEtran}
\bibliography{References}

\vspace{12pt}

\end{document}